\documentclass[11pt]{amsart}
\usepackage{graphicx,amssymb,amsmath,amsthm}
\usepackage{enumerate}
\usepackage{comment,cite,color}
\usepackage{cite,color}
\usepackage{algorithmic}
\usepackage{mathrsfs}
\usepackage[ruled]{algorithm}
\usepackage{epsfig}
\usepackage{lscape}

\usepackage{algorithm}  
\usepackage{algorithmic} 

\newtheorem{theorem}{Theorem}[section]

\newcommand{\A}{\mathbf A}
\newcommand{\bb}{\mathbf b}
\newcommand{\innerp}[1]{\langle {#1} \rangle}

\newcommand{\abs}[1]{\lvert#1\rvert}

\newcommand{\R}{{\mathbb R}}
\newcommand{\Z}{{\mathbb Z}}

\newcommand{\C}{{\mathbb C}}

\renewcommand{\eqref}[1]{(\ref{#1})}
\newcommand{\inner}[1]{\langle #1 \rangle}
\newcommand{\shsp}{\hspace{1em}}
\newcommand{\mhsp}{\hspace{2em}}

\newcommand{\rank}{{\rm rank}}
\newcommand{\supp}{{\rm supp}}

\newcommand{\tr}{{\rm tr}}

\newcommand{\vx}{{\mathbf x}}
\newcommand{\vy}{{\mathbf y}}

\newcommand{\va}{{\mathbf a}}

\newcommand{\vv}{{\mathbf v}}
\newcommand{\vu}{{\mathbf u}}
\newcommand{\vz}{{\mathbf z}}
\newcommand{\vw}{{\mathbf w}}
\newcommand{\vf}{{\mathbf f}}
\newcommand{\vg}{{\mathbf g}}

\renewcommand{\H}{{\mathbb H}}

\newcommand{\MM}{\mathbf M}

\newcommand{\ba}{\mathbf a}

\newcommand{\ii}{{\rm i}}

\newtheorem{prop}{Proposition}[section]
\newtheorem{lem}[prop]{Lemma}

\date{}

\begin{document}

\title{Phase Retrieval From the Magnitudes of Affine Linear Measurements}

\author{Bing Gao}
\address{LSEC, Inst.~Comp.~Math., Academy of
Mathematics and System Science,  Chinese Academy of Sciences, Beijing, 100091, China}
\email{gaobing@lsec.cc.ac.cn}

\author{Qiyu Sun}
\thanks{Qiyu Sun is partially supported by National Science Foundation (DMS-1412413)
    }
\address{Department of Mathematics
             University of Central Florida
Orlando, FL 32816,  USA
}
\email{qiyu.sun@ucf.edu}

\author{Yang Wang}
\thanks{Yang Wang was supported in part by the Hong Kong Research Grant Council grant 16306415 and 16317416 as well as the AFOSR grant FA9550-12-1-0455.}
\address{Department of Mathematics, the Hong Kong University of Science and Technology, Clear Water Bay, Kowloon, Hong Kong}
\email{yangwang@ust.hk}

\author{Zhiqiang Xu}
\thanks{Zhiqiang Xu was supported  by NSFC grant ( 11422113, 11021101, 11331012) and by National Basic Research Program of China (973 Program 2015CB856000)}
\address{LSEC, Inst.~Comp.~Math., Academy of
Mathematics and System Science,  Chinese Academy of Sciences, Beijing, 100091, China}
\email{xuzq@lsec.cc.ac.cn}


\begin{abstract}
In this paper, we consider the phase retrieval problem in which one aims to recover a signal from the magnitudes of affine measurements. Let $\{\ba_j\}_{j=1}^m \subset \H^d$ and $\bb=(b_1, \ldots, b_m)^\top\in\H^m$, where $\H=\R$ or $\C$. We say $\{\ba_j\}_{j=1}^m$ and $\bb$ are affine phase retrievable for $\H^d$ if any $\vx\in\H^d$ can be recovered from the magnitudes of the affine measurements $\{\abs{\innerp{\va_j,\vx}+b_j},\, 1\leq j\leq m\}$. We develop general framework for affine phase retrieval and prove necessary and sufficient conditions for $\{\ba_j\}_{j=1}^m$ and $\bb$ to be affine phase retrievable. We establish results on minimal measurements and generic measurements for affine phase retrieval as well as on sparse affine phase retrieval. In particular, we also highlight some notable differences between affine phase retrieval and the standard phase retrieval in which one aims to recover a signal $\vx$ from the magnitudes of its linear measurements. In standard phase retrieval, one can only recover $\vx$ up to a unimodular constant, while affine phase retrieval removes this ambiguity. We prove that unlike standard phase retrieval, the affine phase retrievable measurements $\{\ba_j\}_{j=1}^m$ and $\bb$ do not form an open set in $\H^{m\times d}\times \H^m$. Also in the complex setting, the standard phase retrieval requires  $4d-O(\log_2d)$ measurements, while the affine phase retrieval only needs $m=3d$ measurements.
\end{abstract}

\maketitle

\bigskip \medskip

\section{Introduction}
\setcounter{equation}{0}
\subsection{Phase retrieval}

Phase retrieval is an active topic of research in recent years as it arises in many different areas of studies (see e.g. \cite{BCE06,BoHa15,CSV12,CEHV13, CCSW, CCD, FMW14,HMW13} and the references therein). For a vector (signal) $\vx\in \H^d$, where $\H=\R$ or $\C$, the aim of phase retrieval is to recover $\vx$ from $\abs{\innerp{\ba_j,\vx}},\, j=1,\ldots,m$, where $\ba_j\in \H^d$ and we usually refer to $\{\ba_j\}_{j=1}^m$ as the {\em measurement vectors}. Since for any unimodular $c\in \H$, we have $\abs{\innerp{\ba_j,\vx}}=\abs{\innerp{\ba_j,c\vx}}$, the best outcome phase retrieval can achieve is to recover $\vx$ up to a unimodular constant.

We briefly overview some of the results in phase retrieval and introduce some notations. For the set of measurement vectors $\{\ba_j\}_{j=1}^m$, we set $\A:=(\ba_1,\ldots,\ba_m)^\top\in \H^{m\times d}$ which we shall refer to as the {\em measurement matrix}. We shall in general identify the set of measurement vectors $\{\ba_j\}_{j=1}^m$ with the corresponding measurement matrix $\A$, and often use the two terms interchangeably whenever there is no confusion.
Define the map $\MM_{\A}:\H^d\rightarrow \R^m_+$ by
\[
\MM_{\A}(\vx)\,\,=\,\, (\abs{\innerp{\ba_1,\vx}},\ldots, \abs{\innerp{\ba_m,\vx}}).
\]
We say $\A$ is {\em phase retrievable }  for $\H^d$ if $\MM_{\A}(\vx)=\MM_{\A}(\vy)$ implies  $\vx\in \{c \vy: c\in \H, \abs{c}=1\}$. There have been extensive studies of phase retrieval from various different angles. For example many efficient algorithms to recover $\vx$ from $\MM_\A(\vx)$ have been developed, see e.g. \cite{CSV12,CESV12, WF,PN13} and their references. One of the fundamental problems on the theoretical side of phase retrieval is the following question: {\em How many
vectors in the measurement matrix $\A$ are needed so that $\A$ is phase retrievable?} It is shown in \cite{BCE06} that for $\A$ to be phase retrievable for $\R^d$, it is necessary and sufficient  that $m \geq 2d-1$.

In the complex case $\H=\C$, the same question becomes much more challenging, however. The minimality question remains open today. Balan, Casazza and Edidin \cite{BCE06}  first show that $\A$ is phase retrievable if it contains $m\geq 4d-2$ generic vectors in $\C^d$. Bodmann and Hammen \cite{BoHa15} show that $m=4d-4$ measurement vectors are possible for phase retrieval through construction (see also Fickus, Mixon, Nelson and Wang \cite{FMW14}). Bandeira, Cahill, Mixon and Nelson \cite{BCMN} conjecture that (a) $m\geq 4d-4$ is necessary for $\A$ to be phase retrievable and, (b) $\A$ with $m\geq 4d-4$ generic measurement vectors is  phase retrievable.
Part (b) of the conjecture is proved by Conca,  Edidin,   Hering and  Vinzant \cite{CEHV13}.
They also confirm  part (a) for the case where $d$ is in the form of $2^k+1, \, k\in \Z_+$.  However, Vinzant in \cite{small} presents a phase retrievable $\A$ for $\C^4$ with $m=11=4d-5<4d-4$ measurement vectors, thus disproving the conjecture. The measurement vectors in the counterexample are obtained using Gr\"obner basis and algebraic computation.

\subsection{Phase retrieval from magnitudes of affine linear measurements}

Here we consider the affine phase retrieval problem, where instead of being given the magnitudes of linear measurements, we are given the magnitudes of affine linear measurements that include shifts. More precisely, instead of recovering $\vx$ from $\{|\innerp{\ba_j,\vx}|\}_{j=1}^m$, we consider recovering $\vx$ from the absolute values of the affine linear measurements
\[
\abs{\innerp{\ba_j,\vx}+b_j},\quad j=1,\ldots,m,
\]
where $\ba_j\in \H^d$, $\bb=(b_1,\ldots,b_m)^\top\in \H^m$. Unlike in the classical phase retrieval, where $\vx$ can only be recovered up to a unimodular constant, we will show that one can recover $\vx$ {\em exactly} from $(\abs{\innerp{\ba_1,\vx}+b_1},\ldots,\abs{\innerp{\ba_m,\vx}+b_m})$ if the vectors $\ba_j$ and shifts $b_j$ are properly chosen.

Let $\A = (\ba_1,\ldots,\ba_m)^\top\in \H^{m\times d}$ and $\bb\in\H^m$.
Define the map $\MM_{\A,\bb}:\H^d\rightarrow \R^m_+$ by
\begin{equation}  \label{eq:M1}
\MM_{\A,\bb}(\vx)=\left(\abs{\innerp{\ba_1,\vx}+b_1},\ldots, \abs{\innerp{\ba_m,\vx}+b_m}\right).
\end{equation}
We say the pair $(\A,\bb)$ (which can also be viewed as a matrix in $\H^{m \times (d+1)}$) is {\em affine phase retrievable} for $\H^d$,  or simply {\em phase retrievable} whenever there is no confusion, if $\MM_{\A,\bb}$ is injective on $\H^d$. Note that sometimes it is more convenient to consider the map
\begin{equation}  \label{eq:M2}
       \MM^2_{{\bf A},{\bf b}}(\vx):=(|\innerp{\ba_1,\vx} +b_1|^2, \ldots, |\innerp{\ba_m,\vx} +b_m|^2).
\end{equation}
Clearly $(\A,\bb)$ is affine phase retrievable if and only if $\MM^2_{\A,\bb}$ is injective on $\H^d$. The goal of this paper is to develop a framework of affine phase retrieval.

There are several motivations for studying affine phase retrieval. It arises naturally in holography, see e.g. \cite{Lie03}. It could also arise in other phase retrieval applications, such as reconstruction of signals in a shift-invariant space from their phaseless samples \cite{CCSW}, where some entries of $\vx$ might be known in advance. In such scenarios, assume that the object signal is $\vy\in \H^{d+k}$ and the first $k$ entries of $\vy$ are known. We can write
$\vy=(y_1,\ldots,y_k, \vx)$, where $y_1,\ldots,y_k$ are known and $\vx\in \H^d$. Suppose that ${\tilde \ba}_j=(a_{j1},\ldots,a_{jk},\ba_j)\in \H^{d+k},\, j=1,\ldots,m$ are the measurement vectors. Then
\[
    \abs{\innerp{{\tilde \ba}_j,\vy}}\,\,=\,\,\abs{\innerp{\ba_j,\vx}+b_j},
\]
where $b_j:=a_{j1}y_1+\cdots+a_{jk}y_k$. So if $(y_1,\ldots,y_k)$ is a nonzero vector, we can take advantage of knowing the first $k$ entries and reduce the standard phase retrieval in $\H^{d+k}$ to affine phase retrieval  in $\H^d$.

\subsection{Our contribution }

This paper considers affine phase retrieval for both real and complex signals.
In Section 2, we  consider the real case $\H=\R$ and prove several necessary and sufficient conditions under which $\MM_{\A,\bb}$ is injective on $\R^d$. For an index set $T\subset \{1,\ldots,m\}$, we use $\A_T$ to denote the sub-matrix $\A_T:=(\ba_j: j\in T)^\top$ of $\A$. Let $ \#T $ denote the cardinality of $ T $, ${\rm span}(\A_T)\subset \R^{\#T}$ denote the subspace  spanned by the column vectors of $\A_T$. In particularly, we show that $(\A,\bb)$  is  affine phase retrievable  for $\R^d$ if and only if ${\rm span}\{\ba_j:j\in S^c\}=\R^d$ for any index set $S\subset \{1,\ldots,m\}$ satisfying $\bb_S\in {\rm span}(\A_S)$.  Based on this result, we prove that the measurement vectors set $\A$ must have at least $m\geq 2d$ elements for  $(\A,\bb)$ to be affine phase retrievable. Furthermore, we prove any generic $\A\in \R^{m\times d}$ and $\bb\in\R^m$, where  $m\geq 2d$ will be affine phase retrievable. The recovery of sparse signals from phaseless measurement also attracts much attention recently \cite{WaXu14,GWaXu16}. In this section, we consider the real affine phase retrieval  for sparse vectors.

We turn to the complex case $\H=\C$  in Section 3. First we establish equivalent necessary and sufficient conditions for $(\A,\bb)$ to be affine phase retrievable for $\C^d$. Using these conditions, we show that $(\A,\bb)\in \C^{m\times (d+1)}$ is {\em not} affine phase retrievable for $\C^d$ if  $m<3d$.  The result is sharp as we also construct an affine phase retrievable $(\A,\bb)$ for $\C^d$ with $m=3d$. This result shows that the nature of affine phase retrieval can be quite different from that of the standard  phase retrieval in the complex setting, where it is known that  $4d -O(\log_2d)$ measurements are needed for phase retrieval \cite{HMW13,WaXu16}.

Note that for $ j=1,\ldots,m$ we have
\begin{equation}\label{eq:shiftphase}
    \abs{\innerp{\ba_j,\vx}+b_j}\,\,=\,\,\abs{\innerp{{\tilde \ba}_j,{\tilde \vx}}}, ~~\mbox{where}~~
      	{\tilde \vx}=\begin{pmatrix} \vx \\ 1\end{pmatrix},
  		{\tilde \ba}_j = \begin{pmatrix} \ba_j \\ b_j\end{pmatrix}.
\end{equation}
It shows that affine phase retrieval for $\vx$ can be reduced to the classical phase retrieval for ${\tilde \vx}\in \C^{d+1}$ from $\abs{\innerp{{\tilde \ba}_j,{\tilde \vx}}}, \,j=1,\ldots,m$. Because the last entry of $\tilde \vx$ is $1$, it allows us to recover $\vx$ without the unimodular constant ambiguity. Observe also from \cite{CEHV13} that  $4(d+1)-4=4d$ generic measurements are enough to recover $\tilde \vx$ up to a unimodular constant, and hence they are also enough to recover $\vx$. In Section 3, we  prove the stronger result that  a generic $(\A,\bb)$ in $\C^{m\times (d+1)}$  with $m\geq 4d-1$ is affine phase retrievable.

The classical phase retrieval has the property that the set of phase retrievable $\A\in \H^{m\times d}$ is an open set, and hence the phase retrievable property is stable under small perturbations \cite{RaduSt, BaWaSt}. Surprisingly, viewing $(\A,\bb)$ as an element in $\H^{m\times (d+1)}$, we prove that the set of affine phase retrievable $(\A,\bb)$ is {\em not} an open set. As far as stability of affine phase retrieval  is concerned, we prove several new results in Section 4. For the standard phase retrieval, one uses $\min_{\abs{\alpha}=1}\|\vx-\alpha \vy\|$ to measure the distance between $\vx$ and $\vy$. The robustness of phase retrieval is established via the lower bound of the following bi-Lipschitz type inequalities for any phase retrievable $\A$,
\begin{equation}\label{eq:stability}
   c\min_{\alpha\in \C, \abs{\alpha}=1}\|\vx-\alpha \vy\|\,\leq\, \|\MM_\A(\vx)-\MM_\A(\vy)\|
   \,\leq\, C\min_{\alpha\in \C, \abs{\alpha}=1}\|\vx-\alpha \vy\|,
\end{equation}
where $c, C>0$ depends only on $\A$ \cite{CCD}. More explicit estimate of the constant $c$ was given in\cite{BaWaSt}. For the affine phase retrieval,  we use $\|\vx-\vy\|$ to measure the distance between $\vx$ and $\vy$ because it is possible to recover $\vx$ exactly in the affine phase retrieval. For the affine phase retrieval, we show that both $\MM_{\A,\bb}$ and $\MM^2_{\A,\bb}$ are bi-Lipschitz continuous on any compact sets, but are not bi-Lipschitz on $\H^d$.

\section{Affine Phase Retrieval for Real Signals}
\setcounter{equation}{0}

We consider affine phase retrieval  of real signals in this section. Several equivalent conditions for affine phase retrieval  are established. We also study affine phase retrieval for sparse signals. In particular we answer the minimality question, namely what is the smallest number of measurements needed for affine phase retrievability for $\R^d$.

\subsection{Real affine phase retrieval}
Let $T\subset \{1,2, \ldots, m\}$. We first recall that for the measurement matrix ${\bf A}=(\ba_1,\ldots,\ba_m)^\top\in \R^{m\times d}$, we use $\A_T$ to denote the submatrix of $\A$ consisting only those rows indexed in $T$, i.e. $\A_T:=(\ba_j: j\in T)^\top$. Similarly we use $\bb_T$ to denote the sub-vector of $\bb$ consisting only entries indexed in $T$. For any matrix ${\bf B}$, we use ${\rm span} ({\bf B})$ to denote the subspace spanned by the {\em columns} of ${\bf B}$. Thus for any index subset $T$, the notation ${\rm span}({\bf A}_T)$ denotes the subspace of $\R^{\# T}$ spanned by the columns of $\A_T$.

\begin{theorem}\label{le:real}
   Let  ${\bf A}=(\ba_1,\ldots,\ba_m)^\top\in \R^{m\times d}$ and ${\bb}=(b_1, \ldots, b_m)^\top\in \R^m$. Then the followings are equivalent:
   \begin{itemize}
   \item[\rm (A)] $(\A,\bb)$ is affine phase retrievable  for $\R^d$.
   \item[\rm (B)] The map $\MM^2_{{\bf A},{\bf b}}$ is injective on $\R^d$, where $\MM^2_{{\bf A},{\bf b}}$ is defined in (\ref{eq:M2}).
   \item[\rm (C)] For any $\vu,\vv\in\R^d$ and $\vu\neq 0$, there exists a $k$ with $1\leq k \leq m$ such that
      $$
           \innerp{\ba_k,\vu}\bigl(\innerp{\ba_k,\vv} +b_k\bigr) \neq 0.
      $$
   \item[\rm (D)] For any $S\subset \{1,2, \ldots, m\}$, if ${\bf b}_S\in {\rm span}({\bf A}_S)$ then ${\rm span} (\A_{S^c}^\top)=  {\rm span}\{\ba_{j}:j\in S^c\}=\R^d$.
   \item[\rm (E)] The Jacobian $J(\vx)$ of the map $\MM^2_{{\bf A},{\bf b}}$ has rank $d$ for all $\vx\in\R^d$.
   \end{itemize}
\end{theorem}
\begin{proof}   The equivalence of (A) and (B) have already been discussed earlier. We focus on the other conditions.

\noindent
(A) $\Leftrightarrow$ (C).~~Assume that $\MM_{{\bf A},{\bf b}}({\bf x})=\MM_{{\bf A},{\bf b}}({\bf y})$ for some ${\bf x} \neq {\bf  y}$ in $\R^d$. For any $j$, we have
\begin{equation} \label{eq:deng}
    \abs{\innerp{\ba_j,{\bf x}}+b_j}^2-\abs{\innerp{\ba_j,{\bf y}}+b_j}^2 =\innerp{\ba_j, {\bf x}-{\bf y}}(\innerp{\ba_j,{\bf x}+{\bf y}}+2b_j).
\end{equation}
Set $2\vu = \vx-\vy$ and $2\vv = \vx+\vy$. Then $\vu \neq 0$ and for all $j$,
\begin{equation} \label{eq:innerp_zero}
      \innerp{\ba_j,\vu}\bigl(\innerp{\ba_j,\vv} +b_j\bigr) = 0.
\end{equation}
Conversely, assume that (\ref{eq:innerp_zero}) holds for all $j$. Let $\vx,\vy\in\R^d$ be given by $\vx-\vy=2\vu$ and $\vx+\vy =2\vv$.
Then $\vx\neq \vy$. However, we would have $\MM^2_{{\bf A},{\bf b}}({\bf x})=\MM^2_{{\bf A},{\bf b}}({\bf y})$ and hence $(\A,\bb)$ cannot be affine phase retrievable.

\noindent
(C) $\Leftrightarrow$ (D).~~Assume that (C) holds. If for some $S\subset \{1,2, \ldots, m\}$ with ${\bf b}_S\in {\rm span}({\bf A}_S)$, we have ${\rm span}\{\ba_{j}: ~j\in S^c\}\neq \R^d$, then we can find $\vu \neq 0$ such that $\innerp{\ba_j,\vu} = 0$ for all $j\in S^c$. Moreover, since ${\bf b}_S\in {\rm span}({\bf A}_S)$, we can find $\vv\in\R^d$ such that $-b_j = \innerp{\ba_j,\vv} $ for all $j\in S$. Thus for all $1 \leq j \leq m$, we have
$$
        \innerp{\ba_j,\vu}\bigl(\innerp{\ba_j,\vv} +b_j\bigr) = 0.
$$
This is a contradiction. The converse clearly also holds.

\noindent
(C) $\Leftrightarrow$ (E).~~ Note that the Jacobian $J(\vv)$ of the map $\MM^2_{{\bf A},{\bf b}}$ at the point $\vv\in\R^d$ is precisely
$$
       J(\vv) = \Bigl((\innerp{\ba_1,\vv} +b_1)\ba_1, (\innerp{\ba_2,\vv} +b_2)\ba_2, \ldots, (\innerp{\ba_m,\vv} +b_m)\ba_m\Bigr),
$$
i.e. the $j$-th column of $J(\vv)$ is $(\innerp{\ba_j,\vv} +b_j)\ba_j$. Thus $\rank(J(\vv))\neq d$ if and only if there exists a nonzero $\vu\in\R^d$ such that
$$
    \vu^\top J(\vv)=\Bigl(\innerp{\ba_1,\vu}\bigl(\innerp{\ba_1,\vv} +b_1\bigr), \ldots, \innerp{\ba_m,\vu}\bigl(\innerp{\ba_m,\vv} +b_m\bigr)\Bigr)=0.
$$
The equivalence of (C) and (E) now follows.
\end{proof}

As an application of Theorem \ref{le:real}, we show that the
minimal number of affine measurements to recover all $d$-dimensional real signals is $2d$.

\begin{theorem}\label{realminimal.thm1}
Let  ${\bf A}=(\ba_1,\ldots,\ba_m)^\top\in \R^{m\times d}$ and ${\bf b}\in \R^m$. If $m\leq 2d-1$, then $(\A,\bb)$ is not affine phase retrievable  for $\R^d$.
\end{theorem}
\begin{proof} We divide the proof into two cases.

\noindent
{\em Case 1}:  ${\rm rank}({\bf A})\leq d-1$.

In this case,  there exists a nonzero vector ${\bf u}\in \R^d$ such that
$ \langle {\bf a}_j, {\bf u}\rangle =0, \ 1\le j\le m$.
Thus for any ${\bf x}\in \R^d$,
$$
    \abs{\innerp{\ba_j,{\bf x}}+b_j}^2=\abs{\innerp{\ba_j,{\bf x}+{\bf u}}+b_j}^2, \quad 1\le j\le m.
$$
This means that  $\MM_{\A,\bb}$ is not injective.

\noindent
{\em Case 2}:  ${\rm rank}({\bf A})= d$.

In this case, there exists an index set $S_0\subset \{1,\ldots,m\}$ with cardinality $d$ so that the square matrix ${\bf A}_{S_0}$ has full rank $d$, which implies
\begin{equation} \label{realminimal.thm1.pf.eq1}
     {\bf b}_{S_0}\in {\rm span}({\bf A}_{S_0}).
\end{equation}
In other words, there exists $\vv \in \R^d$ such that $\innerp{\ba_j,\vv}+b_j=0$ for all $j \in S_0$. Now since $m\le 2d-1$ and $\# S_0=d$, we have $\#S_0^c=m-d\leq d-1$. Hence there exists a nonzero $\vu\in\R^d$ such that $\vu \perp \{\ba_j:~j\in S_0^c\}$. The non-injectivity follows immediately from Theorem \ref{le:real} (C).
\end{proof}

We next consider generic measurements. There are various ways one can define the meaning of being generic. A rigorous definition involves the use of Zariski topology. In this paper, we adopt a simpler definition. An element $\vu\in \H^N$ is generic, if $\vu \in X$ for some dense open set $X$ in $\H^N$ such that $X^c$ is a null set. Sometimes in actual proofs, we obtain the stronger result where $X^c$ is a real algebraic variety. The following theorem on generic measurements also shows that the lower bound given in Theorem \ref{realminimal.thm1} is optimal.

\begin{theorem}\label{realminimal.thm2}
    Let $ m\geq 2d$. Then  a generic $(\A,\bb) \in \R^{m\times (d+1)}$  is affine phase retrievable.
\end{theorem}
\begin{proof}
The theorem follows readily from Theorem \ref{le:real} (D). Note that for a generic $\A \in \R^{m\times d}$, any $d$ rows are linearly independent, so that ${\rm span} (\A_{S^c}^\top) = \R^d$ as long as $\# S^c \geq d$. On the other hand, ${\rm span}({\bf A}_S)$ is a $d$ dimensional subspace in $\R^{\# S}$ and so $\bb_S \not \in {\rm span}({\bf A}_S)$ if $\# S >d$. Thus if ${\bf b}_S\in {\rm span}({\bf A}_S)$, then $\# S \leq d$, which implies $\# S^c \geq d$. Consequently ${\rm span}\{{\bf a}_j:~j\in S^c\}={\rm span} (\A_{S^c}^\top)=\R^d$. Hence $({\bf A}, {\bf b})$ is affine phase retrievable.
\end{proof}

The following theorem highlights a difference between the classical linear phase retrieval and the affine phase retrieval.

\begin{theorem}\label{th:nonopenReal}
    Let $ m\geq 2d$. Then the set of affine phase retrievable $(\A,\bb) \in \R^{m\times (d+1)}$  is not an open set.
\end{theorem}
\begin{proof} We only need to find an affine phase retrievable $(\A,\bb) \in \R^{m\times (d+1)}$ such that for each $\epsilon>0$, there is a small perturbation $(\A',\bb) \in \R^{m\times (d+1)}$ with $\|\A-\A'\|_F <\epsilon$ such that $(\A',\bb)$ is not affine phase retrievable, where $\|\cdot\|_F$ denotes the $l^2$-norm (Frobenius norm). We first do so for $m=2d$. Set
$$
     \A=(I_d, I_d)^\top, \mhsp \bb=(b_{11}, \ldots, b_{d1}, b_{12}, \ldots, b_{d2})^\top.
$$
Here we require that $b_{j1} \neq b_{j2}$ for all $j$ and specially suppose $ b_{12}=0 $. Then $(\A,\bb)$ is affine phase retrievable. To see this, assume that $\vx,\vy \in\R^d$ such that $\MM_{\A,\bb}(\vx)= \MM_{\A,\bb}(\vy)$. Then for each $j$, we must have $|x_j+b_{jk}|=|y_j+b_{jk}|$ for $k=1,2$. Since $b_{j1} \neq b_{j2}$, we must have $x_j=y_j$. Thus $\MM_{\A,\bb}$ is injective and hence $(\A,\bb)$ is phase retrievable.

Now let $\delta>0$ be sufficiently small. We perturb $\A$ to
\begin{equation} \label{eq:perturbReal}
     \A' = \left(I_d+b_{11}\delta E_{21},  I_d\right)^\top,
\end{equation}
where $E_{ij}$ denotes the matrix with the $(i,j)$-th entry being 1 and all other entries being 0. Now set $\vx=(b_{11}, -1/\delta, 0,\ldots,0)^\top $ and $\vy = (-b_{11}, -1/\delta, 0, \ldots, 0)^\top$. It is easy to see that
$$
|\A'\vx+\bb|=|\A'\vy+\bb|.
$$
Hence $(\A',\bb)$ is not affine phase retrievable. By taking $\delta$ sufficiently small, we will have $\|\A'-\A\|_F\leq  \epsilon$. It follows that for $m=2d$, the set of affine phase retrievable $(\A,\bb) \in \R^{m\times (d+1)}$  is not an open set.

In general for $m> 2d$, we can simply take the above construction $(\A,\bb)\in \R^{2d\times (d+1)}$ and augment it to a matrix $(\tilde \A, \tilde\bb)\in \R^{m\times (d+1)}$ by appending $m-2d$ rows of zero vectors to form its last $m-2d$ rows. The $(\tilde\A, \tilde\bb)$ is clearly affine phase retrievable, and the same perturbation above applied to the first $2d$ rows of $\A$ now breaks the affine phase retrievability. Thus for any $m\geq 2d$, the set of affine phase retrievable $(\A,\bb) \in \R^{m\times (d+1)}$  is not an open set.
\end{proof}

\subsection{Real sparse affine phase retrieval}
Set
\[
\Sigma_s(\H^d)\,\,:=\,\, \{\vx\in \H^d: \|\vx\|_0\leq s\}.
\]
We say that $(\A,\bb)\in \H^{m\times (d+1)}$ is {\em $s$-sparse affine phase retrievable for $\H^d$} if $\MM_{\A,\bb}$ is injective on $\Sigma_s(\H^d)$. In this subsection, we show that the minimal number of affine measurements to recover all $s$-sparse real signals is $2s+1$.

\begin{theorem}
\begin{itemize}
\item[\rm (i)]~ Let $1 \leq s \leq d-1$ and $(\A,\bb)\in \R^{m\times (d+1)}$ be $s$-sparse affine phase retrievable for $\R^d$. Then $ m\geq 2s+1$.
\item[\rm (ii)]~ Let $m \geq 2s+1$ and $(\A,\bb)$ be a generic element in $\R^{m\times (d+1)}$. Then  $({\bf A}, {\bf b})$ is $s$-sparse affine phase retrievable for $\R^d$.
\end{itemize}
\end{theorem}

\begin{proof}  (i) We first show that if $(\A,\bb)$ is $s$-sparse affine phase retrievable， then $m \geq 2s+1$. First we claim that the rank of ${\bf A}$ is at least $r=\min(d,2s)$. Indeed, suppose that the claim is false. Then there exists a nonzero vector ${\bf x}\in \Sigma_{r}(\R^d)$, such that ${\bf A} {\bf  x}={\bf 0}$. Write ${\bf x}={\bf u}-{\bf v}$ with ${\bf u},{\bf v}\in \Sigma_s(\R^d)$. Then ${\bf u}\ne {\bf v}$ and  ${\bf A}{\bf u}={\bf A}{\bf v}$. Hence for all $1 \leq j \leq m$, we have
$$
    |\langle {\bf a}_j, {\bf u}\rangle +b_j|=|\langle {\bf a}_j, {\bf v}\rangle+b_j|,
$$
which is a contradiction. Thus $\rank(\A) \geq r=\min(d,2s)$.

Assume that $ m \leq 2s$. We derive a contradiction. Since $s<d$, it follows that $r \geq s+1$. Thus there exists an index set $T\subset \{1, 2, \ldots, m\}$ with $\# T = s+1$, such that $\rank(\A_{T})=s+1$. Without of loss of generality we may assume that $T=\{1, 2, \ldots, s+1\}$. Moreover, we may also without of loss of generality assume that the first $s+1$ columns of $\A_{T}$ are linearly independent. In other words, the $(s+1)\times (s+1)$ submatrix of $\A$ restricted to the first $s+1$ rows and columns is nonsingular. Call this matrix $B$. It follows that there exists a $\vy\in\R^{s+1}$ such that $B\vy = -\bb_{T}$.  Write $\vy=(y_1, \ldots, y_{s+1})^\top$ and set
$$
    \vv_0 = (y_1, \ldots, y_{s+1}, 0, \ldots, 0)^\top\in\R^{d}.
$$
Then $\A_T\vv_0 =-\bb_T$.

If $y_j=0$ for some $1 \leq j \leq s+1$, say $y_{s+1}=0$, we let $\vu=(u_1, \ldots, u_s,0, \ldots, 0)^\top$. Since $\# T^c = m-(s+1) \leq s-1$, there exists such a $\vu_0 \neq 0$ such that $\innerp{\ba_j, \vu_0}=0$ for all $j\in T^c$. Now for $\vx = \vv_0+\vu_0$ and $\vy = \vv_0-\vu_0$, we have $\MM^2_{\A,\bb}(\vx) = \MM^2_{\A,\bb}(\vy)$ and $\vx\neq \vy$. Furthermore, $\vx,\vy \in\Sigma_s(\R^d)$. This is a contradiction. Hence $y_j \neq 0$ for all $1 \leq j \leq s$.

Now for any $1 \leq j_1<j_2 \leq s+1$ consider
\begin{equation}  \label{eq:u-constrain}
      \vu_{j_1,j_2} =(  u_1, \ldots, u_{s+1}, 0, \ldots, 0)^\top\in\R^{d}, \mhsp  u_{j_1} = ty_{j_1},~u_{j_2} = -ty_{j_2}.
\end{equation}
We view the other $u_j$'s and $t$ as unconstrained variables, so there are $s$ variables. Since $\# T^c = m-(s+1) \leq s-1$, it follows that there exists a $\tilde\vu_{j_1,j_2} \neq 0$ satisfying (\ref{eq:u-constrain}) such that $\innerp{\ba_j, \tilde\vu_{j_1,j_2}}=0$ for all $j\in T^c$. If $t\neq 0$, then we may normalize $\tilde\vu_{j_1,j_2}$ so that $t=1$. Set $\vx = \vv_0+\tilde\vu_{j_1,j_2}$ and $\vy = \vv_0-\tilde\vu_{j_1,j_2}$. It follows that $\MM^2_{\A,\bb}(\vx) = \MM^2_{\A,\bb}(\vy)$ and
$$
     \supp(\vx) \subset \{1, 2,\ldots, s+1\}\setminus\{j_2\}, \shsp \supp(\vy) \subset \{1, 2,\ldots, s+1\}\setminus\{j_1\}.
 $$
This is a contradiction.

To complete the proof of $m \geq 2s+1$, we finally need to consider the case that $t=0$ in $\tilde\vu_{j_1,j_2} \neq 0$ for any two indices $1 \leq j_1<j_2 \leq s+1$. But if so, it implies that any $s-1$ columns among the first $s+1$ columns of $\A_{T^c}$ are linearly dependent. In particular, it means the $(m-s-1)\times s$ submatrix of $\A_{T^c}$ restricted to the first $s+1$ columns has rank at most $s-2$. Now because the $(s+1)\times (s+1)$ submatrix of $\A$ restricted to the first $s+1$ rows and columns is nonsingular, we may without loss of generality assume that $s\times s$ submatrix of $\A$ restricted to the first $s$ rows and columns is nonsingular, for otherwise we can make a simple permutation of the indices. The key now is to observe that $(\A,\bb)$ is not $s$-sparse affine phase retrievable, because $(\A,\bb)$ restricted to the first $s$ columns is not affine phase retrievable for $\R^s$. To see this, let $\A'$ be the submatrix of $\A$ consisting of only the first $s$ columns of $\A$. We show $(\A',\bb)$ is not affine phase retrievable for $\R^s$. Note that for $S=\{1,2,\ldots,s\}$, we have $\bb_S\in {\rm span}\,(\A'_S)$ because by assumption $\A'_S$ is nonsingular. But we also know that the rows of $\A_{S^c}$ do not span $\R^s$ because it has $\rank(\A_{S^c}) \leq s-1$. Hence $(\A',\bb)$ is not affine phase retrievable by Theorem \ref{le:real} (D). This completes the proof of $m \geq 2s+1$.

\medskip

(ii)  ~Next we prove for  $m \geq 2s+1$, a generic $(\A,\bb) \in\R^{m \times (d+1)}$ is $s$-sparse affine phase retrievable. The set of all $(\A,\bb)\in \R^{m\times (d+1)}$ has real dimension $m(d+1)$. The goal is to show that the the set of $(\A,\bb)$ that are not $s$-sparse affine phase retrievable  lies in a finite union of subsets of dimension strictly less than $m(d+1)$. Our result then follows.

   For any subset of indices $I,J\subset \{1,\ldots,m\}$ with $\#I ,\#J\leq s$, we say $(\A,\bb) \in \R^{m\times (d+1)}$ is not $(I,J)$-sparse affine phase retrievable  if there exist $\vx\neq \vy$ in $\R^d$ such that
\begin{equation}  \label{eq:IJsparseReal}
    {\rm supp}(\vx)\subset I, \shsp {\rm supp}(\vy)\subset J, \shsp \mbox{and}\shsp \MM^2_{\A,\bb}(\vx) = \MM^2_{\A,\bb}(\vy).
\end{equation}
Let $ \mathcal{A}_{I,J} $ denote the set of all 4-tuples $(\A, \bb, \vx,\vy)$ satisfying (\ref{eq:IJsparseReal}) and $\vx\neq \vy$. Then
$$
      \mathcal{A}_{I,J} \subset \R^{m\times (d+1)}\times \R^{\# I}\times \R^{\# J}.
$$
Then $ \mathcal{A}_{I,J} $ is a well-defined real quasi-projective variety (\!\!\!\cite[Page 18]{alge}). Write $\A=(\ba_1, \ba_2, \ldots, \ba_m)^\top$ and $\bb=(b_1, \ldots, b_m)^\top$. Then by (\ref{eq:deng}), $\MM^2_{\A,\bb}(\vx) = \MM^2_{\A,\bb}(\vy)$ is equivalent to
\begin{equation} \label{eq:IJ-deng}
  \innerp{\ba_j, \vx-\vy}(\innerp{\ba_j,{\bf x}+{\bf y}}+2b_j) = 0, \mhsp j=1,2, \ldots, m.
\end{equation}
Fix any $j$, the above equation holds if and only if
$$
   \innerp{\ba_j, \vx-\vy}=0 \mhsp \mbox{or} \mhsp \innerp{\ba_j,\vx+\vy}+2b_j = 0.
$$
Thus for any $\vx\neq \vy$, the first condition requires $\ba_j$ to lie on a hyperplane, which has co-dimension 1, while the second condition fixes $b_j$ to be $-\innerp{\ba_j,{\bf x}+{\bf y}}/2$. Overall, for any given $\vx \neq \vy$, these two conditions constraint the $j$-th row of $(\A,\bb)$ to lie on a real projective variety of codimension 1. We shall use $X_j(\vx,\vy)$ to denote this variety. Now let $\pi_2: {\mathcal A}_{I,J} \longrightarrow \R^d \times \R^d$ be the projection $(\A,\bb,\vx,\vy) \mapsto (\vx,\vy)$  onto the last two coordinates. Then for any $\vx_0 \neq \vy_0$ in $\R^d$, we have
$$
      \pi_2^{-1}\{(\vx_0,\vy_0)\} = X_1(\vx_0,\vy_0) \times X_2(\vx_0,\vy_0) \times \ldots \times X_m(\vx_0,\vy_0)\times \{\vx_0\}\times \{\vy_0\}.
$$
Hence the dimension of $\pi_2^{-1}\{(\vx_0,\vy_0)\}$ is
$$
      \dim\left(\pi_2^{-1}\{(\vx_0,\vy_0)\} \right) = m(d+1) - m = md.
$$
It follows that $\dim({\mathcal A}_{I,J}) \leq md +\# I+\# J \leq md +2s$.

We now let $\pi_1:  {\mathcal A}_{I,J} \longrightarrow \R^{m\times (d+1)}$ be the projection $(\A,\bb,\vx,\vy) \mapsto (\A,\bb)$. Since projections cannot increase the dimension of a variety, we know that
$$
      \dim\left(\pi_1( {\mathcal A}_{I,J})\right) \leq md +2s = m(d+1) +2s-m<m(d+1).
$$
However, $\pi_1( {\mathcal A}_{I,J} )$ contains precisely those $(\A,\bb)$ in $\R^{m\times (d+1)}$ that are not $(I,J)$-sparse affine phase retrievable. Thus a generic $(\A,\bb)\in\R^{m\times (d+1)}$ is $(I,J)$-sparse affine phase retrievable.

Finally, there are only finitely many indices subsets $I,J$. Hence a generic $(\A,\bb)\in\R^{m\times (d+1)}$ ($m\geq 2s+1$) is $(I,J)$-sparse affine phase retrievable  for any $I, J$ with $\# I, \# J \leq s$. The theorem is proved.
\end{proof}

\section{Affine Phase Retrieval for Complex Signals}
\setcounter{equation}{0}
In this section, we consider affine phase retrieval for complex signals.
Affine phase retrieval for complex signals, like in the case of the classical phase retrieval, poses additional challenges.

\subsection{Complex affine phase retrieval}
We first establish the analogous of Theorem \ref{le:real} for complex signals.

\begin{theorem}\label{th:PRScomplex}
   Let  ${\bf A}=(\ba_1,\ldots,\ba_m)^\top\in \C^{m\times d}$ and ${\bb}=(b_1, \ldots, b_m)^\top\in \C^m$. Then the followings are equivalent:
   \begin{itemize}
   \item[\rm (A)] $(\A,\bb)$ is affine phase retrievable  for $\C^d$.
   \item[\rm (B)] The map $\MM^2_{{\bf A},{\bf b}}$ is injective on $\C^d$.
   \item[\rm (C)] For any $\vu,\vv\in\C^d$ and $\vu\neq 0$, there exists a $1\leq k \leq m$ such that
      $$
          \Re\Bigl( \innerp{\vu, \ba_k}\bigl(\innerp{\ba_k,\vv} +b_k\bigr)\Bigr) \neq 0.
      $$
   \item[\rm (D)] Viewing $\MM^2_{{\bf A},{\bf b}}$ as a map $\R^{2d} \longrightarrow \R^m$, its (real) Jacobian $J(\vx)$ has rank $2d$ for all $\vx\in\R^{2d}$.
   \end{itemize}
\end{theorem}
\begin{proof}   The equivalence of (A) and (B) have already been discussed earlier. We focus on the other conditions.

\noindent
(A) $\Leftrightarrow$ (C).~~Assume that $\MM^2_{{\bf A},{\bf b}}({\bf x})=\MM^2_{{\bf A},{\bf b}}({\bf y})$ for some ${\bf x} \neq {\bf  y}$ in $\C^d$. Observe that for any $a,b\in\C$, we have $|a|^2-|b|^2 = \Re((\bar a-\bar b)(a+b))$. Thus for any $j$, we have
\begin{equation} \label{eq:diff}
    \abs{\innerp{\ba_j,{\bf x}}+b_j}^2-\abs{\innerp{\ba_j,{\bf y}}+b_j}^2  =
       \Re\Bigl( \innerp{ {\bf x}-{\bf y}, \ba_j}(\innerp{\ba_j,{\bf x}+{\bf y}}+2b_j)\Bigr).
\end{equation}
Set $2\vu = \vx-\vy$ and $2\vv = \vx+\vy$. Then $\vu \neq 0$ and for all $j$,
\begin{equation} \label{eq:C_innerp_zero}
      \Re\Bigl( \innerp{ \vu, \ba_j}(\innerp{\ba_j,\vv}+b_j)\Bigr)=0.
\end{equation}
Conversely, assume that (\ref{eq:C_innerp_zero}) holds for all $j$. Let $\vx,\vy\in\C^d$ be given by $\vx-\vy=2\vu$ and $\vx+\vy =2\vv$. Then $\vx\neq \vy$. However, we would have $\MM^2_{{\bf A},{\bf b}}({\bf x})=\MM^2_{{\bf A},{\bf b}}({\bf y})$ and hence $(\A,\bb)$ cannot be affine phase retrievable.

\noindent
(C) $\Leftrightarrow$ (D).~~The $k$-th entry of $\MM^2_{{\bf A},{\bf b}}({\bf x})$ is $|\inner{\ba_k, \vx}+b_k|^2$. Since all variables here are complex, we shall separate them into the real and imaginary parts by adopting the notation $\vx = \vx_R+i\vx_I$, $\ba_k =\ba_{k,R}+i\ba_{k,I}$ and $b_k = b_{k,R}+ib_{k,I}$. The $k$-th entry of $\MM^2_{{\bf A},{\bf b}}({\bf x})$ is now
$$
       |\inner{\ba_k, \vx}+b_k|^2 = \left(\innerp{\ba_{k,R},\vx_R} + \innerp{\ba_{k,I},\vx_I} + b_{k,R}\right)^2
       				+  \left(\innerp{\ba_{k,R},\vx_I} - \innerp{\ba_{k,I},\vx_R} - b_{k,I}\right)^2.
$$
It follows that the (real) Jacobian of $\MM^2_{{\bf A},{\bf b}}(\vx_R,\vx_I)$ is
$$
   J(\vx):=J(\vx_R,\vx_I) = 2\begin{pmatrix}   \ba_{1,R}^\top\cdot \alpha_1(\vx)-\ba_{1,I}^\top\cdot \beta_1(\vx)
                                                      &\ba_{1,I}^\top\cdot \alpha_1(\vx)+\ba_{1,R}^\top\cdot \beta_1(\vx)\\
                                 \ba_{2,R}^\top\cdot \alpha_2(\vx)-\ba_{2,I}^\top\cdot \beta_2(\vx)
                                      &\ba_{2,I}^\top\cdot \alpha_2(\vx)+\ba_{2,R}^\top\cdot \beta_2(\vx)\\
                                        \vdots & \vdots \\
                              \ba_{m,R}^\top\cdot \alpha_m(\vx)-\ba_{m,I}^\top\cdot \beta_m(\vx)
                                  & \ba_{m,I}^\top\cdot \alpha_m(\vx)+\ba_{m,R}^\top\cdot \beta_m(\vx)
\end{pmatrix},
$$
where $\alpha_j(\vx):=\innerp{\ba_{j,R},\vx_R}+\innerp{\ba_{j,I},\vx_I}+b_{j,R}$ and $\beta_j(\vx):=\innerp{\ba_{j,R},\vx_I}-\innerp{\ba_{j,I},\vx_R}-b_{j,I}$ for all $0 \leq j\leq m$.

Now assume that $\rank (J(\vx))$ is not $2d$ everywhere. Then there exist $\vv=\vv_R+i\vv_I$ and $\vu=\vu_R+i\vu_I \neq 0$, such that $\vu$ as a vector in $\R^{2d}$ is in the null space of $J(\vv)$, i.e.,
$$
     J(\vv) \begin{pmatrix} \vu_R\\ \vu_I\end{pmatrix} = 0.
$$
It follows that for all $1 \leq k \leq m$, we have
\begin{equation} \label{null_u}
     C_k :=\innerp{\ba_{k,R},\vu_R} \alpha_k(\vv)-\innerp{\ba_{k,I},\vu_R} \beta_k(\vv)
               + \innerp{\ba_{k,I},\vu_I} \alpha_k(\vv)+\innerp{\ba_{k,R},\vu_I} \beta_k (\vv) =0.
\end{equation}
But one can readily check that $C_k$ is precisely
$$
     C_k=\Re\Bigl( \innerp{ \vu, \ba_k}(\innerp{\ba_k,\vv}+b_k)\Bigr).
$$
Thus $(\A,\bb)$ cannot be affine phase retrievable by (C).

The converse clearly also holds.
Assume that (C) is false. Then there exists $\vv,\vu\in\C^d$ and $\vu\neq 0$ such that
$$
    \Re\Bigl( \innerp{ \vu, \ba_k}(\innerp{\ba_k,\vv}+b_k)\Bigr) = 0
$$
for all $1 \leq k \leq m$. It follows that (\ref{null_u}) holds for all $k$ and hence
$$
     J(\vv) \begin{pmatrix} \vu_R\\ \vu_I\end{pmatrix} = 0.
$$
Thus $\rank(J(\vv)) < 2d$.
\end{proof}

\vspace{3mm}

\subsection{Minimal measurement number}

We now show that the minimal number of measurements needed to be affine phase retrievable is $ 3d$. This is  surprising compared to the classical affine phase retrieval, where the minimal number is $4d-O(\log_2d)$.

\begin{lem}\label{le:3}
    Let $z_1, z_2\in \C$. Suppose that $b_1,b_2,b_3\in \C$ are not collinear on the complex plane.
 Then $z_1=z_2$ if and only if   $\abs{z_1+b_j}=\abs{z_2+b_j},\, j=1,2,3$.
\end{lem}
\begin{proof}
We use $z_{j,R}$ and $z_{j,I}$ to denote the real and imaginary part of $z_j$, and similarly for  $b_{j,R}$ and $b_{j,I}$. Assume the lemma is false, and that there exist $z_1, z_2$ with $z_1\neq z_2$ so that  $\abs{z_1+b_j}^2=\abs{z_2+b_j}^2,\, j=1,2,3$.
  Note that $\abs{z_1+b_j}^2=\abs{z_2+b_j}^2$ implies that
\begin{equation}\label{eq:ling}
     (z_{2,R}-z_{1,R})\cdot b_{j,R}+(z_{2,I}-z_{1,I})\cdot b_{j,I}=\frac{|z_1|^2-|z_2|^2}{2}, \quad j=1,2,3.
\end{equation}
The (\ref{eq:ling}) together with $z_1\ne z_2$ implies that $b_1,b_2,b_3$ are collinear. This is a contradiction.
\end{proof}

\vspace{3mm}

\begin{theorem}\label{th:m3d}
\begin{itemize}
     \item[\rm (i)]  ~Suppose that $(\A,\bb)\in \C^{m\times (d+1)}$ is affine phase retrievable in $\C^d$. Then $m\geq 3d$.
     \item[\rm (ii)]   ~Let $B:=(\ba_1,\ldots,\ba_d)\in \C^{d\times d}$ be nonsingular. Set $\A=(B,B,B)^\top\in \C^{3d\times d}$. Let
     $$
           \bb=(b_{11},\ldots,b_{d1},b_{12},\ldots,b_{d2},b_{13},\ldots,b_{d3})^\top\in \C^{3d}
     $$
      such that  $b_{j1},b_{j2},b_{j3}$  are not collinear in $\C$ for any $1\leq j\leq d$. Then $(\A,\bb)$ is affine phase retrievable in $\C^d$.
\end{itemize}
\end{theorem}
\begin{proof}
(i) ~~Write $\A=(\ba_1,\ldots,\ba_m)^\top\in \C^{m\times d}$. Assume that $m <3d$. Clearly $\rank(\A)=d$, for otherwise we will have $\A\vx=0$ for some $\vx\neq 0$ and hence $\MM_{\A,\bb}(\vx) = \MM_{\A,\bb}(0)$. Hence there exists a $T\subset \{1, \ldots, m\}$ with $\# T=d$ such that $\rank(A_T)=d$, which means we can find $\vv\in\C^d$ such that
$$
      \innerp{\ba_k, \vv} + b_k = 0, \mhsp k \in T.
$$
Now, because $\# T^c =m-d <2d$, and the system of homogeneous linear equations for the variable $\vu$ with $\vv$ fixed,
$$
     \Re\Bigl( \innerp{ \vu, \ba_j}(\innerp{\ba_j,\vv}+b_j)\Bigr)=0, \mhsp j\in T^c
$$
has $2d$ real variables $\vu_R, \vu_I$, it must have a nontrivial solution. The two vectors $\vu\neq 0, \vv$ combine to yield
$$
     \Re\Bigl( \innerp{ \vu, \ba_j}(\innerp{\ba_j,\vv}+b_j)\Bigr)=0
$$
for all $1 \leq j \leq m$. This contradicts with (C) in Theorem \ref{th:PRScomplex}.
\vspace{2mm}

(ii)~~To prove $(\A,\bb)$ is affine phase retrievable, we prove that $\MM_{\A,\bb}(\vx)=\MM_{\A,\bb}(\vy)$ implies $\vx=\vy$ in $\C^d$. The property $\MM_{\A,\bb}(\vx)=\MM_{\A,\bb}(\vy)$ implies that
\begin{equation}\label{eq:bjdeng}
    \abs{\innerp{\ba_j,\vx}+b_{jk}}\,\,=\,\,\abs{\innerp{\ba_j,\vy}+b_{jk}} , \mhsp j=1,\ldots,d,\shsp k=1,2,3.
\end{equation}
Thus by Lemma \ref{le:3}, for each fixed $j$ we have
\[
     \innerp{\ba_j,\vx}=\innerp{\ba_j,\vy}.
\]
This implies $\vx=\vy$ since the matrix $B=(\ba_1,\ldots,\ba_d)$ is nonsingular.
\end{proof}

\vspace{3mm}

It is well known that in the classical phase retrieval, the set of all phase retrievable $\A \in \C^{m\times d}$ is an open set in $ \C^{m\times d}$. But for affine phase retrieval, as with the real affine phase retrieval case, this property no longer holds. The following theorem shows that this property also doesn't  hold in the complex case when $m \geq 3d$.

\begin{theorem} \label{th:ins}
      Let $m \geq 3d$. Then the set of affine phase retrievable $(\A,\bb) \in \C^{m\times (d+1)}$ is not an open set in $ \C^{m\times (d+1)}$. In fact, there exists an affine phase retrievable $(\A,\bb)\in \C^{3d\times (d+1)}$, which satisfies the conditions in  Theorem \ref{th:m3d} (ii). Given any $\epsilon>0$, there exists  $ (\A',\bb) \in \C^{3d\times (d+1)} $  which does not have affine phase retrievable property such that
 \[
  \|\A'-\A\|_F\leq  \epsilon,
 \]
 where $\|\cdot \|_F$ denotes the Frobenius norm.
\end{theorem}
\begin{proof}
Following the construction given in  Theorem \ref{th:m3d} (ii), we set $\A=(B,B,B)^\top$, where $B$ is nonsingular and
$$
           \bb:=(\underbrace{\ii,\ldots,\ii}_d,\underbrace{0,\ldots,0}_d,\underbrace{1,\ldots,1}_d)^\top\in \C^{3d}.
$$
We will show that there exists an arbitrarily small perturbation $\A'$ such that $(\A', \bb)$ is no longer affine phase retrievable. Making a simple linear transformation $\vx = B^{-1}\vy$, we see that all we need is to show that this property holds for $\A=(I_d, I_d, I_d)^\top$, where $I_d$ is the $d\times d$ identity matrix.  Let $\delta>0$ be sufficiently small. We perturb $\A$ to
\begin{equation} \label{perturb}
     \A' = (I_d+\ii\delta E_{21},  I_d, I_d)^\top,
\end{equation}
where $E_{21}$ denotes the matrix with the $(2,1)$-th entry being 1 and all other entries being 0. Now set $\vx=(\ii,-1/\delta,0,\ldots,0)^\top$ and $\vy = (-\ii, -1/\delta, 0, \ldots, 0)^\top$. It is easy to see that
$$
|\A'\vx+\bb|=|\A'\vy+\bb|.
$$
Thus $(\A',\bb)$ is not affine phase retrievable. By taking $\delta$ sufficiently small we will have $\|\A'-\A\|_F\leq  \epsilon$.

In general for $m> 3d$, like the real case,  we can simply take the above construction $(\A,\bb)\in \C^{3d\times (d+1)}$ and augment it to a matrix $(\tilde \A, \tilde\bb)\in \C^{m\times (d+1)}$ by appending $m-3d$ rows of zero vectors to form its last $m-3d$ rows. $(\tilde\A, \tilde\bb)$ is clearly affine phase retrievable, and the same perturbation above applied to the first $3d$ rows of $\tilde\A$ now breaks the affine phase retrievability.

Thus for any $m\geq 3d$, the set of affine phase retrievable $(\A,\bb) \in \C^{m\times (d+1)}$  is not an open set.
\end{proof}
\medskip

We next consider complex affine phase retrieval  for generic measurements. We have the following theorem:


\begin{theorem}\label{th:4d1}
   Suppose that $m\geq 4d-1$. Then a generic $(\A,\bb) \in \C^{m \times (d+1)} $ is affine phase retrievable in $\C^d$.
\end{theorem}

\begin{proof}
   Let $N=m+1$. Then $N \geq 4d = 4(d+1)-4$. Hence by \cite{CEHV13}, there is an open dense set of full measure $X\subset \C^{N \times (d+1)} $, such that any ${\bf F}\in X $ is linear phase retrievable in the classical sense. Write ${\bf F}=(\vf_1, \vf_2, \ldots, \vf_N)^\top$, where each $\vf_j \in \C^{d+1}$. For each $\vg\in\C^{d+1}$, denote $X_{\vg} :=\{{\bf F}=(\vf_1, \vf_2, \ldots, \vf_N)^\top \in X: ~\vf_N = \vg\}$. Then there exists a $\vg_0\in\C^{d+1}$, such that the projection of $X_{\vg_0}$ onto $\C^{(N-1)\times (d+1)}$ with the last row removed is a dense open set with full measure. Thus ${\bf F}=(\vf_1, \ldots, \vf_{N-1},\vg_0)^\top$ is  phase retrievable  in $\C^{d+1}$ in the classical sense for a generic $(\vf_1, \ldots, \vf_{N-1})^\top \in \C^{(N-1)\times (d+1)}$.

   Now let $P_0\in\C^{(d+1)\times (d+1)}$ be nonsingular such that $P_0\vg_0 = {\bf e}_{d+1}$. Then for any ${\bf F}\in X_{\vg_0}$, we have
$$
       {\bf G}:={\bf F}P_0^\top =(P_0\vf_1, \ldots, P_0\vf_{N-1},{\bf e_{d+1}})^\top
       =: (\vg_1, \ldots, \vg_{N-1},{\bf e}_{d+1})^\top.
$$
It is linear phase retrievable in the classical sense for generic $\vg_1, \ldots, \vg_{N-1}$.  In particular, any vector $\vy =(x_1, \ldots, x_d, 1)^\top$ can be recovered by $\abs{{\bf G}\vy}$, where $\abs{\cdot}$ means the entry-wise absolute value. However, note that the last entry of $\vy$ is 1, and the last row of ${\bf G}$ is ${\bf e}_{d+1}^\top$. So the measurement from last row provides no information. In other words, the above $\vy$ can be recovered exactly from the measurements provided by the first $N-1 = m$ rows of ${\bf G}$. This means precisely that the first $N-1 = m$ rows of ${\bf G}$ are affine phase retrievable. Let $(\A, \bb)$ denote the first $m$ rows of ${\bf G}$. It follows that $(\A,\bb)\in \C^{m \times (d+1)} $ is affine phase retrievable.  Therefore a generic  $(\A,\bb) \in \C^{m \times (d+1)} $ ($ m\geq 4d-1$) is affine phase retrievable.
\end{proof}

\subsection{Sparse complex affine phase retrieval}

We now focus on sparse affine phase retrieval  by proving that generic $(\A,\bb)$ is $s$-sparse affine phase retrievable  if $m\geq 4s+1$.

\begin{theorem}\label{th:4k1}
Let $m \geq 4s+1$. Then a generic $(\A,\bb) \in \C^{m\times (d+1)}$ is $s$-sparse affine phase retrievable .
\end{theorem}
\begin{proof}
    The proof here is very similar to the proof in the real case. The set of all $(\A,\bb)$ has real dimension $\dim_\R(\C^{m\times (d+1)})=2m(d+1)$. The goal is to show that the the set of $(\A,\bb)$'s that are not $s$-sparse affine phase retrievable  lies in a finite union of subsets, each of which is a projection of real hypersurfaces of dimension strictly less than $2m(d+1)$. This would yield our result.

   For any subset of indices $I,J\subset \{1,\ldots,m\}$ with $\#I ,\#J\leq s$, we say $(\A,\bb) \in \C^{m\times (d+1)}$ is not $(I,J)$-sparse affine phase retrievable  if there exist $\vx\neq \vy$ in $\C^d$ such that
\begin{equation}  \label{eq:IJsparse}
    {\rm supp}(\vx)\subset I, \shsp {\rm supp}(\vy)\subset J, \shsp \mbox{and}\shsp \MM_{\A,\bb}^2(\vx) = \MM_{\A,\bb}^2(\vy).
\end{equation}
Let $ \mathcal{A}_{I,J} $ denote the set of all 4-tuples $(\A, \bb, \vx,\vy)$ satisfying (\ref{eq:IJsparse}) and $\vx\neq \vy$. Then
$$
      \mathcal{A}_{I,J} \subset \C^{m\times (d+1)}\times \C^{\# I}\times \C^{\# J},
$$
where we view $(\A,\bb)$ as an element of $ \C^{m\times (d+1)}$.
For our proof we shall identify $\C^{m\times (d+1)}\times \C^{\# I}\times \C^{\# J}$ with $\R^{m\times 2(d+1)}\times \R^{2\# I}\times \R^{2\# J}$. In this case $ \mathcal{A}_{I,J} $ is a well-defined real quasi-projective variety (\!\!\!\cite[Page 18]{alge}). Note that $\MM_{\A,\bb}^2(\vx) = \MM_{\A,\bb}^2(\vy)$ yields $|\innerp{\ba_j, \vx}+b_j|^2 = |\innerp{\ba_j, \vy}+b_j|^2$ for all $1\leq j\leq m$, where $\A=(\ba_1, \ba_2, \ldots, \ba_m)^\top$ and $\bb=(b_1, \ldots, b_m)^\top$. By (\ref{eq:diff}), this is equivalent to
\begin{equation} \label{eq:IJ-diff}
   \Re\Bigl( \innerp{ {\bf x}-{\bf y}, \ba_j}(\innerp{\ba_j,{\bf x}+{\bf y}}+2b_j)\Bigr) = 0, \mhsp j=1,2, \ldots, m.
\end{equation}
Fix any $j$, the above equation holds if and only if
\begin{itemize}
\item $\innerp{ {\bf x}-{\bf y}, \ba_j}=0$; {\rm or }
\item $\innerp{ {\bf x}-{\bf y}, \ba_j}\neq 0$ but (\ref{eq:IJ-diff}) holds.
\end{itemize}
Thus for any $\vx\neq \vy$, the first condition requires $\ba_j$ to lie on a hyperplane, which has real co-dimension 2, while the second condition requires $b_j$ to be on a line in $\C$ (depending on $\vx,\vy,\ba_j$). Overall, for any given $\vx \neq \vy$, these two conditions constraint the $j$-th row of $(\A,\bb)$ to lie on a real projective variety of codimension 1.
We shall use $X_j(\vx,\vy)$ to denote this variety. Now let $\pi_2: {\mathcal A}_{I,J} \longrightarrow \C^d \times \C^d$ be the projection $(\A,\bb,\vx,\vy) \mapsto (\vx,\vy)$ onto the last two coordinates. Then for any $\vx_0 \neq \vy_0$ in $\C^d$, we have
$$
   \pi_2^{-1}\{(\vx_0,\vy_0)\} = X_1(\vx_0,\vy_0) \times X_2(\vx_0,\vy_0) \times \ldots \times X_m(\vx_0,\vy_0)\times \{\vx_0\}\times \{\vy_0\}.
$$
Hence the real dimension of $\pi_2^{-1}\{(\vx_0,\vy_0)\}$ is
$$
      \dim_\R\left(\pi_2^{-1}\{(\vx_0,\vy_0)\} \right) = 2m(d+1) - m = 2md + m.
$$
It follows that $\dim_\R({\mathcal A}_{I,J}) \leq 2md+m +2\# I+2\# J \leq 2md+m +4s$.

We now let $\pi_1:  {\mathcal A}_{I,J} \longrightarrow \C^{m\times (d+1)}$ be the projection $(\A,\bb,\vx,\vy) \mapsto (\A,\bb)$. Since projections cannot increase the dimension of a variety, we know that
$$
      \dim_\R\left(\pi_1( {\mathcal A}_{I,J})\right) \leq 2md+m +4s = 2m(d+1) +4s-m<2m(d+1).
$$
However, $\pi_1( {\mathcal A}_{I,J})$ contains precisely those $(\A,\bb)$ in $\C^{m\times (d+1)}$ that are not $(I,J)$-sparse affine phase retrievable. Thus a generic $(\A,\bb)\in\C^{m\times (d+1)}$ is $(I,J)$-sparse affine phase retrievable.

Finally, there are only finitely many indices subsets $I,J$. Hence a generic $(\A,\bb)\in\C^{m\times (d+1)}$ ($m\geq 4s+1$) is $(I,J)$-sparse affine phase retrievable  for any $I, J$ with $\# I, \# J \leq s$. The theorem is proved.
\end{proof}

\section{Stability and Robustness of Affine Phase Retrieval}
\setcounter{equation}{0}

Stability and robustness are important properties for affine phase retrieval. For the standard  phase retrieval, stability and robustness have been studied in several papers, see e.g. \cite{BaWaSt,BCMN, CCD, GWaXu16}. In this section, we establish stability and robustness results for both maps $\MM_{\A,\bb}$ and $\MM^2_{\A,\bb}$.

\begin{theorem}  \label{th:Stability}
Assume that $(\A,\bb) \in \H^{m \times (d+1)}$ is affine phase retrievable.
 Assume that $\Omega\subset \H^d$ is a   compact set.
Then there exist positive constants $C_1,  C_2, c_1, c_2$ depending on $(\A,\bb)$ and $\Omega$ such that  for any $\vx, \vy \in \Omega $, we have
\begin{align}
     \frac{c_1}{1+\|\vx\|+\|\vy\|} \,\|\vx-\vy\|
     &\leq \left\|\MM_{\A,\bb}(\vx) - \MM_{\A,\bb}(\vy)\right\| \leq  C_1 \|\vx-\vy\|,\label{eq:Lip1}\\
     {c_2}\|\vx-\vy\|
         &\leq \left\|\MM^2_{\A,\bb}(\vx) - \MM^2_{\A,\bb}(\vy)\right\|\leq
          C_2({1+\|\vx\|+\|\vy\|}) \|\vx-\vy\|. \label{eq:Lip2}
\end{align} 
\end{theorem}
\begin{proof}
Write $\A=(\ba_1, \ldots, \ba_m)^\top$ and $\bb = (b_1, \ldots, b_m)^\top$. We first establish the inequality for the map  $\MM_{\A,\bb}^2(\vx)$, where we recall
$$
\MM_{\A,\bb}^2(\vx)=\left(\abs{\innerp{\ba_1,\vx}+b_1}^2,\ldots, \abs{\innerp{\ba_m,\vx}+b_m}^2\right).
$$
Denote the matrix $(\A, \bb) \in \H^{m \times (d+1)}$ by $(\A,\bb) = (\tilde\ba_1, \ldots, \tilde\ba_m)^\top$, where
   ${\tilde \ba}_j:=\left(
\begin{array}{c}
\ba_j \\ b_j
\end{array}
\right), \,j=1,\ldots,m$.
 Similarly we augment $\vx, \vy\in\H^d$ into $\tilde\vx, \tilde\vy\in\H^{d+1}$ by appending 1 to the $(d+1)$-th entry. Now we have
\begin{align*}
   \MM_{\A,\bb}^2(\vx)&=(\abs{\innerp{\ba_1,\vx}+b_1}^2,\ldots, \abs{\innerp{\ba_m,\vx}+b_m}^2)\\
            &=\left(\tr({\tilde{\ba}_1\tilde{\ba}_1^*\tilde \vx\tilde \vx^*}),\ldots, \tr({\tilde {\ba}_m\tilde {\ba}_m^* \tilde \vx\tilde \vx^*})\right)
             =:{\bf T} (\tilde \vx \tilde \vx^*),
\end{align*}
 where  ${\bf T}$ is a linear transformation from $\H^{(d+1)\times (d+1)}$ to $\R^m$.


 Let $X_\Omega =\{\tilde\vx\tilde\vx^*\in\H^{(d+1)\times (d+1)}: \ \vx\in \Omega\}$, 
$$\Theta_\Omega=\{{\bf S} \in\H^{(d+1)\times (d+1)}: \ \|{\bf S}\|_F=1,\  t{\bf S}\in X_\Omega-X_\Omega\ {\rm for \ some} \ t>0\}$$
and
$$ \tilde \Theta_\Omega=\left\{ {\bf S}:=
\left(\begin{array}{cc} \vz \vw^*+\vw \vz^*  & \vz\\
\vz^* & 0\end{array}\right):\    \vz\in \H^d, \vw\in (\Omega+\Omega)/2\ {\rm and}\  \|{\bf S}\|_F=1 \right\}, $$
where $\|\cdot \|_F$ denotes the $l^2$-norm (Frobenius norm) of a matrix.
Then
\begin{equation} \label{eq110}
\Theta_\Omega\subset \tilde \Theta_\Omega
\end{equation}
because
\begin{equation*}
 {\bf S}=t^{-1}  (\tilde \vx \tilde \vx^*-\tilde \vy\tilde \vy^*)=
\left(\begin{array}{cc} \vz \vw^*+\vw \vz^*  & \vz\\
\vz^* & 0\end{array}\right)\in \tilde \Theta_\Omega \ \ {\rm for \ all} \ {\bf S}\in \Theta_\Omega,\end{equation*}
where the existence of $t>0, \vx, \vy\in \Omega$  in the first equality follows from the definition of
$\Theta_\Omega$ and the second equality holds for
 $\vz=(\vx-\vy)/t$ and $\vw=(\vx+\vy)/2$.

For any  ${\bf S}= \left(\begin{array}{cc} \vz \vw^*+\vw \vz^*  & \vz\\
\vz^* & 0\end{array}\right)\in \tilde\Theta_\Omega $,
 we have
\begin{equation} \label{eq111}
{\bf T}({\bf S})  =
{\bf T}\left(\begin{array}{cc} \vx\vx^*-\vy\vy^*  & \vx-\vy\\
\vx^*-\vy^* & 0\end{array}\right) =  {\bf T}(  \widetilde {\vx} \widetilde{\vx}^*-
\tilde {\vy} \tilde{\vy}^*)
 = {\bf  M}_{{\bf A}, {\bf b}}^2(\vx)-{\bf M}_{{\bf A},{\bf  b}}^2(\vy)\ne 0
\end{equation}
by the affine phase retrievability of $({\bf A},{\bf b})$, where
  $\vx= \vw+\vz/2$ and $\vy=\vw-\vz/2$.
Clearly $\tilde \Theta_\Omega$ is a compact set. This together with \eqref{eq111}
implies that
\begin{equation}\label{eq1102}
c_2:=\inf_{{\bf S}\in \tilde \Theta_\Omega}\|{\bf T} ({\bf S})\|>0.
\end{equation}
Therefore
\begin{eqnarray}\label{eq1103}
& &    \left\|\MM_{\A,\bb}^2(\vx)-\MM_{\A,\bb}^2(\vy)\right\|
     =   \|{\bf T}({\tilde \vx}{\tilde \vx}^*-{\tilde \vy}{\tilde \vy}^*)\|\nonumber\\
    & \ge & \Big(\inf_{{\mathbf S}\in \Theta_\Omega} \|{\bf T}({\mathbf S})\|\Big)   \|{\tilde \vx}{\tilde \vx}^*-{\tilde \vy}{\tilde \vy}^*\|_F
    \ge  
 c_2 \|{\tilde \vx}{\tilde \vx}^*-{\tilde \vy}{\tilde \vy}^*\|_F,
\end{eqnarray}
where the last inequality holds by \eqref{eq110}.


Now for the unit vector ${\mathbf e}_{d+1}$, we have
$$
  \|{\tilde \vx}{\tilde \vx}^*-{\tilde \vy}{\tilde \vy}^*\|_F\ge \left\|({\tilde \vx}{\tilde \vx}^*-{\tilde \vy}{\tilde \vy}^*){\mathbf e}_{d+1}\right\|
    =\|\tilde \vx - \tilde\vy\| = \|\vx-\vy\|.
$$
 This, together with \eqref{eq1102}
and \eqref{eq1103},
 establishes the lower bound in (\ref{eq:Lip2}).

Because $\MM^2_{\A,\bb}(\vx)$ is linear in $X={\tilde \vx}{\tilde \vx}^*$, we must also have
$$
   \left\|\MM_{\A,\bb}^2(\vx)-\MM_{\A,\bb}^2(\vy)\right\|
         \leq C_2' \|{\tilde \vx}{\tilde \vx}^*-{\tilde \vy}{\tilde \vy}^*\|_F.
$$
However using the standard estimate, we have
\begin{align*}
     \|{\tilde \vx}{\tilde \vx}^*-{\tilde \vy}{\tilde \vy}^*\|_F
          \leq \|\tilde\vx\|\,\|\tilde\vx-\tilde\vy\| +\|\tilde\vy\|\,\|\tilde\vx-\tilde\vy\|
          \leq 2(1+\|\vx\| +\|\vy\|)\,\|\vx-\vy\|.
\end{align*}
Here we have used the facts that $\|\tilde\vx-\tilde\vy\| =\|\vx-\vy\|$ and $\|\tilde\vx\|\leq 1+\|\vx\|$. Taking $C_2 = 2C_2'$ yields the upper bound in (\ref{eq:Lip2}).

We now prove the inequalities for $\MM_{\A,\bb}$. The upper bound in (\ref{eq:Lip1}) is straightforward. Note that
$$
    \Bigl|\abs{\innerp{\ba_j,\vx}+b_j}-\abs{\innerp{\ba_j,\vy}+b_j}\Bigr|
    \leq |\innerp{\ba_j,\vx-\vy}| \leq \|\ba_j\|\,\|\vx-\vy\|.
$$
It follows that
$$
\left\|\MM_{\A,\bb}(\vx)-\MM_{\A,\bb}(\vy)\right\| \leq \Bigl(\sum_{j=1}^m \|\ba_j\|\Bigr)\,\|\vx-\vy\|.
$$
The upper bound in (\ref{eq:Lip1}) thus follows by letting $C_1 = \sum_{j=1}^m \|\ba_j\|$.

To prove the lower bound, we observe that
\begin{align*}
    \Bigl|\abs{\innerp{\ba_j,\vx}+b_j}^2-\abs{\innerp{\ba_j,\vy}+b_j}^2\Bigr|
    & =\Bigl|\abs{\innerp{\ba_j,\vx}+b_j}-\abs{\innerp{\ba_j,\vy}+b_j}\Bigr|
        (\abs{\innerp{\ba_j,\vx}+b_j}+\abs{\innerp{\ba_j,\vy}+b_j}) \\
    &\leq L(1+\|\vx\| +\|\vy\|) \Bigl|\abs{\innerp{\ba_j,\vx}+b_j}-\abs{\innerp{\ba_j,\vy}+b_j}\Bigr|,
\end{align*}
where $L>0$ is a constant depending only on $(\A,\bb)$. Hence
$$
   \left\|\MM^2_{\A,\bb}(\vx)-\MM^2_{\A,\bb}(\vy)\right\|
   \leq L(1+\|\vx\| +\|\vy\|) \left\|\MM_{\A,\bb}(\vx)-\MM_{\A,\bb}(\vy)\right\|.
$$
It now follows from the lower bound $\|\MM^2_{\A,\bb}(\vx)-\MM^2_{\A,\bb}(\vy)\| \geq c_2\|\vx-\vy\|$ and setting $c_2 = c_1/L$ that
$$
   \left\|\MM_{\A,\bb}(\vx) - \MM_{\A,\bb}(\vy)\right\|
        \geq \frac{c_1}{1+\|\vx\|+\|\vy\|} \,\|\vx-\vy\|.
$$
The theorem is proved.
\end{proof}

\begin{prop} \label{coro:bi-Lip}
  Neither  $\MM_{\A,\bb}$ nor   $\MM^2_{\A,\bb}$  is bi-Lipschitz on $\H^d$.
\end{prop}
\begin{proof}
The map $\MM^2_{\A,\bb}(\vx)$ is not bi-Lipschitz follows from the simple observation that it is quadratic in $\vx$ (more precisely, in $\Re(\vx)$ and $\Im(\vx)$). No quadratic function can be bi-Lipschitz on the whole Euclidean space.

To see $\MM_{\A,\bb}(\vx)$ is not bi-Lipschitz, we fix a nonzero $\vx_0\in \H^d$. Take $\vx=r \vx_0$ and $\vy=-r \vx_0$, where $r>0$. Note that
\[
   \|\MM_{\A,\bb}(\vx)-\MM_{\A,\bb}(\vy)\|=\Bigl(\sum_{j=1}^m (\abs{r \innerp{\ba_j,\vx_0}+b_j}-\abs{r \innerp{\ba_j,\vx_0}-b_j})^2\Bigr)^{1/2}
\]
and
\[
\|\vx-\vy\|\,\,=\,\, 2 r\|\vx_0\|  .
\]
Then
\begin{equation}\label{eq:rt}
\frac{\|\MM_{\A,\bb}(\vx)-\MM_{\A,\bb}(\vy)\|}{\|\vx-\vy\|}=\frac{1}{2\|\vx_0\|}
\Bigl(\sum_{j=1}^m (\abs{ \innerp{\ba_j,\vx_0}+b_j/r}-\abs{ \innerp{\ba_j,\vx_0}-b_j/r})^2\Bigr)^{1/2}.
\end{equation}
A simple observation is that the right side of (\ref{eq:rt}) tending to 0 as $r \rightarrow \infty$.  Hence for any $\delta >0$, we can choose $r$ large enough so that
\[
\frac{\|\MM_{\A,\bb}(\vx)-\MM_{\A,\bb}(\vy)\|}{\|\vx-\vy\|}\leq \delta.
\]
Thus $\MM_{\A,\bb}(\vx)$ is not bi-Lipschitz.
\end{proof}

\bibliographystyle{plain}

\begin{thebibliography}{99}
%
\bibitem{RaduSt}
R. Balan, Stability of phase retrievable frames, available at http://arxiv.org/abs/1308.5465, 2013.

\bibitem{BCE06}
R. Balan, P. Casazza, D. Edidin, On signal reconstruction without phase,
Appl. Comput. Harmon. Anal. 20 (2006) 345-356.
%
\bibitem{BCMN}
A.S. Bandeira, J. Cahill, D.G. Mixon, A.A. Nelson,
Saving phase: Injectivity and stability for phase retrieval,
Appl. Comput. Harmon. Anal. 37(1) (2014) 106-125.

\bibitem{BaWaSt}
R. Balan, Y. Wang, Invertibility and robustness of phaseless reconstruction, Appl. Comput. Harmon. Anal. 38(3) (2015) 469-488.

\bibitem{BoHa15}
B.G. Bodmann, N. Hammen,
Stable phase retrieval with low-redundancy frames,
Adv. Comput. Math. 41(2)(2015) 317-331.
%
\bibitem{CCD}
J. Cahill, P.G. Casazza, I. Daubechies,
Phase retrieval in infinite-dimensional Hilbert spaces, available at http://arxiv.org/abs/1601.06411, 2016.
%
\bibitem{CCSW}
Y. Chen, C. Cheng, Q. Sun, H.C. Wang, Phase retrieval of real signals in a principal shift-invariant space, available at http://arxiv.org/abs/1603.01592, 2016.
%
\bibitem{CESV12}
E.J. Cand\`{e}s, Y. Eldar, T. Strohmer, V. Voroninski, Phase retrieval via matrix completion, SIAM J. Imaging Sci. 6(1) (2013) 199-225.
%
\bibitem{CSV12}
E.J. Cand\`{e}s, T. Strohmer, V. Voroninski, PhaseLift: exact and stable signal recovery from magnitude measurements via convex programming, Comm. Pure Appl. Math.
66(8)(2013) 1241-1274.
%
\bibitem{WF}
E.J. Cand\`{e}s, X.D. Li, M. Soltanolkotabi,Phase retrieval via wirtinger flow: theory and algorithm, IEEE Trans. Inf. Theory 61(4) (2015) 1985-2007.
%
\bibitem{CEHV13}
A. Conca, D. Edidin, M. Hering, C. Vinzant, Algebraic characterization of injectivity in phase retrieval, Appl. Comput. Harmon. Anal.
38(2) (2015) 346-356.

\bibitem{E15}
D. Edidin,
Projections and phase retrieval, Appl. Comput. Harmon. Anal. (2015).
%
\bibitem{FMW14}
M. Fickus, D.G. Mixon, A.A. Nelson, Y. Wang,
Phase retrieval from very few measurements.
Linear Algebra Appl. 449 (2014) 475-499.
%
\bibitem{GWaXu16}
B. Gao, Y. Wang, Z.Q. Xu, Stable signal recovery from phaseless measurements, J. Fourier Anal. Appl. 22(4) (2016) 787-808.

\bibitem{alge}
J. Harris, Algebraic geometry, first ed., Springer-Verlag, New York, 1992.
%
\bibitem{HMW13}
T. Heinosaari, L. Mazzarella, M.M. Wolf, Quantum tomography under prior information,
 Comm. Math. Phys. 318(2) (2013) 355-374.
%
\bibitem{Lie03}
M. Liebling, T. Blu, E. Cuche, P. Marquet, C.D. Depeursinge, M. Unser, Local amplitude and phase retrieval method for digital holography applied to microscopy, Proc. SPIE 5143 (2003) 210-214
%
\bibitem{PN13}
P. Netrapalli, P. Jain, S. Sanghavi, Phase retrieval using alternating minimization, IEEE Trans. Signal Process. 63(18) (2015) 4814-4826.
%
%
\bibitem{small}
C. Vinzant, A small frame and a certificate of its injectivity, available at http://arxiv.org/abs/1502.04656, 2015.
%
%
\bibitem{WaXu14}
Y. Wang, Z.Q. Xu, Phase retrieval for sparse signals,  Appl. Comput. Harmon. Anal., 37(3) (2014) 531-544.

\bibitem{WaXu16}
Y. Wang, Z.Q. Xu, Generalized phase retrieval: measurement number, matrix recovery and beyond, available at http://arxiv.org/abs/1605.08034, 2016.
%

\end{thebibliography}

\end{document}